\documentclass[a4paper,onecolumn,11pt]{quantumarticle}
\pdfoutput=1
\usepackage[utf8]{inputenc}
\usepackage[english]{babel}
\usepackage[T1]{fontenc}
\usepackage{amsmath,amsthm,amsfonts,amssymb}
\usepackage{hyperref}
\usepackage{comment}

\usepackage{tikz}
\usepackage{lipsum}

\newtheorem{theorem}{Theorem}
\newtheorem{fact}{Fact}
\newtheorem{lemma}{Lemma}
\newtheorem{corollary}{Corollary}

\newcommand{\bra}[1]{\langle #1|}
\newcommand{\ket}[1]{|#1\rangle}
\newcommand{\braket}[2]{\langle #1|#2\rangle}
\newcommand{\cent}[0]{\mbox{\textcent}}
\newcommand{\dollar}[0]{\$}
\newcommand{\confsep}{@ @ @ @ }

\newcommand{\val}{\mathsf{val}}
\newcommand{\nxt}{\mathsf{next}}

\newcommand{\myvec}[1]{ \left( \begin{array}{c} #1 \end{array} \right) }
\newcommand{\mymat}[2]{ \left( \begin{array}{#1} #2 \end{array} \right) }
\newcommand{\mypar}[1]{\left( #1 \right)}

\newcommand{\knapsackgame}{\mathtt{KNAPSACK\mbox{-}GAME}}

\begin{document}

\title{QIP $ \subseteq $ AM(2QCFA) } 

\author{Abuzer Yakary{\i}lmaz}
\email{abuzer.yakaryilmaz@lu.lv}
\affiliation{Center for Quantum Computer Science, University of Latvia, R\={\i}ga, LV-1586, Latvia}
\orcid{0000-0002-2372-252X}
\thanks{}
\maketitle

\begin{abstract}
    The class of languages having polynomial-time classical or quantum interactive proof systems ($\mathsf{IP}$ or $\mathsf{QIP}$, respectively) is identical to $\mathsf{PSPACE}$. We show that $\mathsf{PSPACE}$ (and so $\mathsf{QIP}$) is subset of $\mathsf{AM(2QCFA)}$, the class of languages having Arthur-Merlin proof systems where the verifiers are two-way finite automata with quantum and classical states (2QCFAs) communicating with the provers classically. Our protocols use only rational-valued quantum transitions and run in double-exponential expected time. Moreover, the member strings are accepted with probability 1 (i.e., perfect-completeness).
\end{abstract}

\section{Introduction}

Proof systems play a central role in classical and quantum complexity theory \cite{Co93A,Wa09}. In polynomial time computation, interactive proof systems (IPSs) \cite{GMR85,GMR89}, Arthur-Merlin (AM) proof systems \cite{Ba85}, and quantum interactive proof systems \cite{Wat99B} have the same verification power, i.e., they can verify all and only languages in $ \mathsf{PSPACE}$.

In constant-space, IPSs with two-way probabilistic finite state automaton (2PFA) verifiers can verify any language in $ \mathsf{ASPACE(n)} $ (i.e., alternating linear space) \cite{DS92}. On the other hand, AM systems with 2PFA verifiers are much weaker: any language verified by them is in $\mathsf{P}$, and they cannot verify some languages in $\mathsf{L}$ such as $\mathtt{PAL} = \{ w \in \{a,b\}^* | w = w^{reverse} \} $ \cite{Co91,DS92}.

More than a decade ago, Yakary{\i}lmaz introduced and investigated AM system where the verifier is a two-way automaton with quantum and classical states (2QCFA) \cite{Yak12A,Yak13C}. We denote $\mathsf{AM_1(2QCFA})$ as the class of the languages verified by these systems with perfect-completeness (i.e., member strings are accepted by the verifier with probability 1). 
Yakary{\i}lmaz showed that $\mathsf{AM_1(2QCFA})$ contains $ \mathsf{NP} $, $ \mathsf{ASPACE(n)} $, and an $\mathsf{NEXP}$-complete problem. In this paper, we show that $\mathsf{AM_1(2QCFA})$ also contains $\mathsf{PSPACE}$. We use similar techniques given in \cite{Yak12A,Yak13C}.

We give the definitions and notations used throughout the paper in Section~\ref{sec:preliminaries}. Then, we present several technical details used in our proofs in Section~\ref{sec:technical-details}. After that, we present a protocol for a $\mathsf{PSPACE}$-complete language $\knapsackgame$ in Section~\ref{sec:knapsack-game}. We follow our main result in Section~\ref{sec:linear-space-TMs}. We close the paper with Section~\ref{sec:concluding-remarks}.

We assume the reader is familiar with proof systems, quantum computing, and automata theory. For background on classical and quantum proof systems, we refer to \cite{Co93A,Wa09};  for a comprehensive review of quantum automata, see \cite{Ay15,AY21}; and for further details on AM systems with the verifiers using constant-size quantum memory, see \cite{Yak12A,Yak13C}.

\section{Preliminaries}
\label{sec:preliminaries}

Two-way finite automaton with quantum and classical states (2QCFA) is a classically controlled automaton accessing a finite-size quantum register \cite{AW02}. A 2QCFA can apply unitary or measurement operators to its quantum register, and it can also process any onserved measurement outcome. It is also known as two-way quantum finite automaton with classical head \cite{YS11A}. Here, we define 2QCFAs with superoperators (easily simulated by unitary operators and projective measurements with the help of classical states, e.g., \cite{YS11A}) to shorten our quantum programs and to use only rational transition values. 

A superoperator $\mathcal{E}$ is composed by $k$ operational elements $ \{ E_1, \ldots, E_k \} $ for some $k>0$ satisfying that
\[
    \sum_{j=1}^k E_j^\dagger E_j = I.
\]
Equivalently, the columns of the following matrix form an orthonormal set
\[
    \begin{array}{|c|}
         \hline 
         E_1 \\ \hline
         E_2 \\ \hline
         \vdots \\ \hline
         E_k \\ \hline
    \end{array} ~ .
\]
If $k=1$, then it is a unitary operator. Any measurement operator can also be a represented as a superoperator, where the operator indices are associated with the measurement outcomes.

After applying $\mathcal{E}$ to the quantum state $\ket{v}$, denoted $\mathcal{E}(\ket{v})$, we obtain the following unnormalized state vectors:
\begin{equation}
    \label{eq:unnormalized-list}
    \{ \widetilde{\ket{v_1}}, \ldots, \widetilde{\ket{v_k}}  \},    
\end{equation}
where $ \widetilde{\ket{v_j}} = E_j \ket{v}$. Here the outcome $j$ can be observed with probability $ p_j =  \braket{v_j}{v_j} $, and so $ \widetilde{\ket{v_j}} $ is normalized as $ \ket{v_j} = \frac{ \widetilde{\ket{v_j}} }{ \sqrt{p_j} } $ (if $p_j>0$). Thus, the mixture obtained by $ \mathcal{E}(\ket{v}) $ is represented as
\[
    \left\{  (p_1,\ket{v_1}), \ldots, (p_k,\ket{v_k}) \right\}.
\]
This mixture can be represented as a single mathematical object called density matrix:
\[
    \rho = \sum_{j=1}^k p_j \ket{v_j}\bra{v_j} = \sum_{j=1}^k \widetilde{\ket{v_j}} \widetilde{\bra{v_j}}.
\]

In this paper, though we use superoperators with several operational elements, we keep our ``useful'' computation in one of the pure states after each operation and discard the rest of pure states. Thus, the density matrix is not used.

Arthur-Merlin (AM) proof systems are known as public-coin IPSs, as the computation of verifiers are visible to the provers. A 2QCFA verifier is formally a 8-tuple
\[
    V =(S,Q,\Sigma,\Upsilon,\delta,s_I,q_I,s_{acc},s_{rej}),
\]
where
\begin{itemize}
    \item $S $ is a finite set of classical states, formed by two disjoint sets, $S_c$ and $S_r$ containing the communication and reading states, respectively;
    \item $Q$ is a finite set of quantum states;
    \item $ \Sigma $ is the input alphabet not containing the left and right en-markers ($\cent$ and $\dollar$, respectively);
    \item $ \Upsilon $ is the communication alphabet;
    \item $ \delta $ is the transition function, formed by $\delta_c$, $\delta_r$, and $\delta_q$ responsible for communication, classical control, and quantum transitions, respectively;
    \item $s_I \in S$ is the starting state; and,
    \item $s_{acc},s_{rej} \in S$ are the accepting and rejecting states, respectively. 
\end{itemize}

A given input $w \in \Sigma^*$ is placed on a read-only input tape between two end-makers as $ \cent w \dollar $. The verifier $V$ starts in states $ (s_I,\ket{q_I}) $, and the input head is placed on $\cent$.

The verifier $V$ communicates with a prover through a communication channel storing one symbol from $\Upsilon$ or a measurement outcome (i.e., the index of the corresponding operation element) of superoperators. When in a communication state, say $s_c$, $V$ writes a symbol on the communication channel, and then the prover overwrites it with their respond, say $ \upsilon $. After that $V$ updates its state as $s'$, where $\delta_c(s_c,\upsilon) = s'$.

When in a reading state, say $s_r$, and scanning an input symbol, say $\sigma$, $V$ makes two transitions. First, a superoperator is determined by $ \delta_q (s_r,\sigma) = \mathcal{E} $, and it is applied to the quantum register. The measurement outcome, say $j$, is written on the communication channel to be shared with the prover. Second, $V$ updates its state and head position based on the current classical state, scanning input symbol, and the measurement outcome: $ \delta_r(s_r,\sigma,j) = (s',d) $, where $s'$ is the new classical state; and the input head moves one symbol to the left (right), if $d=-1$ ($d=1$), and, the input head does not move if $d=0$. It must be guaranteed that the input head never leaves $ \cent w \dollar $ on the tape. 

The input $w$ is accepted if $V$ enters $s_{acc}$, and it is rejected if $V$ enters $s_{rej}$. In either case, the computation is terminated.

A language $L \subseteq \Sigma^*$ is said to be verified by $V$ with error bound $\epsilon \in [0,\frac{1}{2}) $ if and only if
\begin{enumerate}
    \item when $w \in L$, $V$ communicates with an honest prover $P$ and accepts $w$ with probability at least $1 - \epsilon$, and,
    \item when $w \notin L$, $V$ rejects $w$ with probability at least $1 - \epsilon$ by communicating with any possible prover $P^*$ (including cheating ones). 
\end{enumerate}
It is also said that $L$ is verified by $V$ with bounded error. If every member is accepted with probability 1, then it is said that $L$ is verified with perfect-completeness.

The class $\mathsf{AM(2QCFA)}$ is the set of languages verified by 2QCFAs with bounded error, and the class $\mathsf{AM_1(2QCFA)}$ is the set of languages verified by 2QCFAs with perfect-completeness. Here we assume that the entries of superoperators are formed by rational numbers. (With arbitrary real-valued transitions, for example, 2QCFAs can verify every language \cite{SY17}.) These classes may also be defined with algebraic or efficiently computable numbers (see \cite{BV97,Wat03}).

\section{Technical details}
\label{sec:technical-details}

\subsection{Rational-valued superoperators}
\label{sec:rational-super}

We do our calculations with rational-valued or integer matrices throughout the paper. More specifically, we solve certain tasks with these matrices and then embed them into some superoperators. 

Our generic technique is as follows. Suppose that we use two rational-valued matrices $ M_1 $ and $ M_2 $. They form our main operation elements after dividing by some integers. Then, we define as many as auxiliary operation elements to have a superoperator: 
\[
    \dfrac{1}{D} ~ \begin{array}{|c|}
         \hline 
         M_1 \\ \hline
         M_2 \\ \hline
         M_3 \\ \hline
         \vdots \\ \hline
         M_k \\ \hline
    \end{array} = \mathcal{E} = 
    \begin{array}{|c|}
         \hline 
         E_1 \\ \hline
         E_2 \\ \hline
         E_3 \\ \hline
         \vdots \\ \hline
         E_k \\ \hline
    \end{array},
\] 
where $1/D$ is the normalization factor, $E_1$ and $E_2$ are the main operation elements, and $E_3,\ldots,E_k$ are auxiliary operation elements for some $k>1$. The value of $D$ and the entries are $M_3, \ldots, M_k$ are determined to ensure that $\mathcal{E}$ is a superoperator. The columns of 
\[
    \begin{array}{|c|}
         \hline 
          M_1 \\ \hline
         M_2 \\ \hline
    \end{array},
\]
may not be orthogonal or may not have the same length. Step by step, by adding the entries below (as the entries of $M_3$, $M_4$, \ldots) we can make the columns are pair-wise orthogonal and then by adding further entries to make the lengths are equal. We refer the reader to \cite{YS11A} for the details. 

We always pick $D$ as some positive integers.  If we have several superoperators with different $D$ values, we add more auxiliary operation elements to ensure that all has the same $D$. 

\subsection{Encoding binary numbers}
\label{sec:binary-encoding}

 We present encoding of binary numbers by using two-dimensional vectors. We start in 
    \[ 
        v_0 = \myvec{1 \\ 0 }. 
    \]
    When reading symbols $0$ and $1$, we multiply it with the following matrices:
    \[
        A_0 = \mymat{cc}{1 & 0 \\ 0 & 2}
        ~~\mbox{and}~~
        A_1 = \mymat{cc}{1 & 0 \\ 1 & 2}.
    \]
    The first entry does not change, and the binary values are encoded on the second entry. We show the correctness of this method by induction.
    
    Let $b$ be a binary number and $\val(b)$ be its value. We start with trivial cases when $b=0$ and $b=1$:
    \[
        v_1  = A_0 \cdot v_0 = \mymat{cc}{1 & 0 \\ 0 & 2} \myvec{1 \\ 0 } = \myvec{1 \\ 0} = \myvec{1 \\ \val(0)}
    \]
    and
    \[
        v_1 = A_1 \cdot v_0 = \mymat{cc}{1 & 0 \\ 1 & 2} \myvec{1 \\ 0 } = \myvec{1 \\ 1} = \myvec{1 \\ \val(1)} .
    \]
    
    Suppose that $b$ with length $i>0$ is encoded as described. That is, we have
    \[
        v_{i} = \myvec{1 \\ \val(b) } .
    \] 
    Then, we read $0$ or $1$:
    \[
        v_{i+1} = A_0 v_{i} = \mymat{cc}{1 & 0 \\ 0 & 2} \myvec{1 \\ \val(b) } = \myvec{1 \\ 2 \cdot \val(b)} = \myvec{1 \\ \val(b0)}
    \]
    or
    \[
        v_{i+1} = A_1 v_{i} = \mymat{cc}{1 & 0 \\ 1 & 2} \myvec{1 \\ \val(b) } = \myvec{1 \\ 2 \cdot \val(b) + 1 } = \myvec{1 \\ \val(b1)}.
    \]
    As easy to see, $b0$ and $b1$ are encoded correctly.

\subsection{Encoding $m$-ary numbers}
\label{sec:m-ary-encoding}

This is a generalization of the previous subsection. We use $ m $ digits for $m$-ary numbers: $\{0,1,\ldots,m-1\}$. When reading $ k $ among them, we use the following encoding matrix:
\[
    A_k = \mymat{cc}{1 & 0 \\ k & m}.
\]
The rest is the same.

\subsection{Coin-flip}
\label{sec:coin-flip}

For a given quantum register, we can use the following superoperator with two operational elements to simulate a fair coin-flip:
\[
    E_0 =  \dfrac{1}{\sqrt{2}} I ~~\mbox{and}~~ E_1 = \dfrac{1}{\sqrt{2}} I,
\]
where $I$ is the identity matrix. Thus, we can observe the outcomes 0 and 1 with equal probability while the quantum state does not change. 

To be consistent with the framework given in Sec.~\ref{sec:rational-super}, we can use a superoperator with $D^2$ operational elements:
\[
     \underbrace{\dfrac{1}{D} I, \ldots, \dfrac{1}{D}}_{D^2~\mbox{times}} I.
\]
Here we can pick any two of them as the main operational elements.

\subsection{Superoperator for the identity operator}
\label{sec:identity}

To have a well-tuned normalization factor in our protocols, we may use the the following superoperator instead of the identity operator:
\[
     \underbrace{\dfrac{1}{D} I, \ldots, \dfrac{1}{D}}_{D^2~\mbox{times}} I.
\]
Here we can pick one of them as the main operational element.

\subsection{Encoding the configurations of a Turing machine}
\label{sec:encode-conf-TM}

Here we explain how to encode the configurations of a single-tape single-head Turing Machine (TM), say $M$, on any given input, say $w$. The tape is two-way infinite indexed by the integers. Any given input $w$ is placed from the left to the right on the tape starting from the tape cell indexed by 1. The rest of cells contain $\#$ symbols at the beginning of the computation. Any configuration of $M$ in any step is of the form $xsy$, where $s$ is the state of $M$, $xy$ is the significant piece of the tape content obtained by removing all insignificant blank symbols except the ones on the left and right, and the tape head is situated on the left-most symbol of $y$. The initial configuration is $ s_{I} \# w \# $, where $s_I$ is the initial state. When in state $s$ and reads symbol $\sigma$, $M$ updates its state with $s'$, overwrites the symbol under the head with $\sigma'$, and updates the head position by at most one cell based on its transition rule $ (s,\sigma) \rightarrow (s',\sigma',d) $, where $d \in \{-1,0\,1\}$ represents of the head moves to the left, stays stationary, or moves to the right, respectively. 

As the tape head can move at most one cell in a single transition, the length of a configuration can be changed by at most 1. Here are two representative examples:
\begin{itemize}
    \item We implement the transition $(s,\#)  \rightarrow (s',\#, -1) $ when in the configuration $s \# 0 0 1 \# $. Thus, the new configuration becomes $ s' \# \# 0 0 1 \# $. The blank symbol under the head is overwritten with itself, and then the head is moved one cell to the left. The length of configuration is increased by 1.
    \item We implement the transition $(s,1)  \rightarrow (s',\#, -1) $ when in the configuration $ \# 10 s 1 \# $. Thus, the new configuration is $ \# 1 s' 0 \# $. The symbol $1$ under the head is overwritten with the blank symbol, and then the head is moved one cell to the left. The length of configuration is decreased by 1. 
\end{itemize}

Moreover, in a single transition, at most three consecutive symbols of a configuration can be changed:
\begin{itemize}
    \item The symbol next to the state can be changed.
    \item The state itself can be changed.
    \item The new state can be swapped with the symbol next to it or before it.
\end{itemize}

Each configuration is represented by a $m$-ary string for some $m>0$, and so, we can encode their values in base-$m$. Here $m$ is the size of the configuration alphabet formed by the tape alphabet and the set of states. Each symbol in the configuration alphabet is associate by a number in $\{0,1,\ldots,m-1\}$, where  
$\#$ is never associated with 0. Thus, none of the configurations represented in base-$m$ starts with 0.

We can encode each configuration by using the integer matrices given in Sec.~\ref{sec:m-ary-encoding}. While reading a configuration from left to right, we can also easily output the valid successor configuration. By reading the first three symbols, we can output the first symbol of the (valid) successor configuration. After that, for each symbol of the configuration, we can output a new symbol of the successor configuration. After reading the whole configuration, we use up to three more steps to complete outputting the successor configuration. While outputting the successor configuration, we can also encode its value in the same way of encoding the value of scanning configuration.

\section{Knapsack game}
\label{sec:knapsack-game}

We give the definition of a Knapsack game shown to be PSPACE-complete under log-space reduction \cite{Travers06,FearnleyJ13,FearnleyJ15}:
\begin{equation}
	\label{eq:knapsackgame}
	 \knapsackgame = \{ S~ \forall(a_1 , b_1) \exists(e_1,f_1) \cdots \forall(a_n , b_n) \exists(e_n,f_n)  \},
\end{equation}
where 
\begin{itemize}
	\item $S$ and each $ a_i,b_i,e_i $, and $f_i$ are natural numbers given in binary ($1 \leq i \leq n$); and,
	\item for every $ x=(x_1,\ldots,x_n) \in \times_{i=1}^n  \{ a_i,b_i \} $, there exists a $ y=(y_1,\ldots,y_n) \in \times_{i=1}^n  \{ e_i,f_i \} $ such that $ S = \sum_{i=1}^n = x_i+y_i $.
\end{itemize}

\begin{theorem}
    \label{thm:knapsack-game}
    The language $\knapsackgame$ can be verified by a 2QCFA with perfect-completeness in exponential expected time. 
\end{theorem}
\begin{proof}
    At the beginning, the verifier reads the input once and deterministically checks if it is in the form of Eq.~\ref{eq:knapsackgame}. If not, the input is rejected and computation is terminated.

    In the remaining part, we assume that the input is in the form of Eq.~\ref{eq:knapsackgame}. The verifier uses a quantum register with 4 states, $\{q_1,q_2,q_3,q_4\}$.
    
    The verifier executes an infinite loop. In each iteration, the verifier reads the input from the left to the right. For each symbol, it applies a superoperator. If the outcome is one of the main operational elements, the computation continues. Otherwise, the current iteration is terminated. Until reading the right end-marker, each superoperator has only a single main operation element. Thus, as long as the iteration is not terminated, the computation continues with a single pure state. In this way, with an exponentially small probability, the computation survives until the end, and then the verifier makes a decision.

    \textbf{Step 0.} The quantum register starts in state $ (1~~0~~0~~0)^T$. By reading the left end-marker, the unnormalized quantum state is set to 
    \[
        N_{\cent} \myvec{1 \\ 0 \\ 0 \\ 1}
    \]
    in the non-terminated part, where $N_{\cent}$ is the normalization factor.
    Then the verifier reads $S$ and attempts to encode its value into the amplitude of $q_2$. If this encoding is successful, the unnormalized quantum state is
    \[
        N_0 \myvec{1 \\ \val(S) \\ 0 \\ 1},
    \]
    where $N_0$ is the cumulative normalization factor so far. For encoding, we embed $A_0$ and $A_1$ given in Sec.~\ref{sec:binary-encoding} into the related main operation elements of superoperators.

    \textbf{Step 1.} The verifier reads the first $\forall$ and then makes a fair coin-flip as described in the second part of Sec.~\ref{sec:coin-flip}. If the outcome is 0, then the verifier encodes the value of $a_1$. Otherwise, it encodes the value of $b_1$. In other words, the verifier picks $x_1 \in \{a_1,b_1\}$. Binary encoding technique is the same as in Step 0 but the verifier uses the amplitude of $q_3$ this time, i.e., the value of $x_1$ is encoded into the amplitude of $q_3$ with the help of the amplitude of $q_1$. Before reading the first $\exists$ symbol, the encoded value is subtracted from $S$ and the amplitude of $q_3$ is set to 0. Besides, the amplitude of $q_4$ is halved. To have a well-tuned normalization factor at the end, the verifier applies the superoperator described in Sec.~\ref{sec:identity} while reading any other symbol.

    After that, the verifier requests from the prover one bit value: If the prover sends $0$ (resp., $1$), the verifier encodes the value of $e_1$ (resp., $f_1$) into the amplitude of $q_3$. We represent the encoded value as $y_1 \in \{e_1,f_1\}$. After that, the value of $y_1$ is subtracted from $S$, and the amplitude of $q_3$ is set to 0 and the amplitude of $q_4$ is halved. As described above, to have a well-tuned normalization factor at the end, the verifier applies the superoperator described in Sec.~\ref{sec:identity} while reading any other symbol.
    
    At this point, unless the iteration is terminated, the unnormalized quantum state, is
    \[
        N_1 \myvec{1 \\ val(S) - x_1 - y_1 \\ 0 \\ 1 / 4},
    \]
    where $N_1$ is the cumulative normalization factor so far.

    \textbf{Step 2.} The verifier reads $\forall (a_2,b_2) \exists (e_2,f_2) $ and do the same as in Step 1. If $x_2 \in \{a_2,b_2\}$ is picked by the verifier and $ y_2 \in \{e_2,f_2\} $ is picked by the prover, the unnormalized quantum state, unless the iteration is terminated, is
    \[
        N_2 \myvec{1 \\ val(S) - x_1 - y_1 - x_2 - y_2 \\ 0 \\ (1/4)^2 },
    \]
    where $N_2$ is the cumulative normalization factor so far.

    \textbf{Steps 3., 4., \dots, $\mathbf{n}$.} We continue in the same way. Let $T$ be the summation of all chosen $x_i$s and $y_i$s:
    \[
        T = ( x_1 + y_1) + (x_2+y_2) +  \cdots + (x_n + y_n).
    \]
    Thus, if the iteration has not been terminated until reading the right end-marker, the quantum state is
    \[
        N_n \myvec{1 \\ S-T \\ 0 \\ (1/4)^n},
    \]
    where $N_n$ is the cumulative normalization factor so far.

    \textbf{Step $\mathbf{n+1}$.} After reading $\dollar$, the input is accepted based on the amplitude of $q_4$, and it is rejected based on the amplitude of $q_2$. Otherwise, the verifier continues with the next iteration. For this, we obtain two unnormalized quantum states as
    \[
        \ket{v_{acc}} = N_{\dollar} \myvec{0 \\ 0 \\ 0 \\ (1/4)^n }
        ~~\mbox{and}~~
        \ket{v_{rej}} = N_{\dollar} \myvec{0 \\ S-T \\ 0 \\ 0 },
    \]
    where $N_{\dollar}$ is an exponentially small value in the input length. If the verifier reads in total $l$ symbols (the input and end-markers) in one iteration, then indeed we have
    \[
        N_{\dollar} =  \left( \frac{1}{D} \right)^{l},
    \]
    which is the same for every sequence of choices by the verifier and prover. 

    \textbf{Decision.} We start with the analysis of the accepting and rejecting probabilities of a single iteration. We have $2^n$ different universal choices. For each of them, the verifier provides existential choices. 

    If the input is a member string, we know that for each universal choices we have
    \[
    S=T
    \]
    by communicating with an honest prover. That means the input is rejected with probability 0.

    Each universal choice leads us to an acceptance condition and it is probability is 
    \[
        N_{\dollar}^2 \left( \dfrac{1}{16} \right)^n.
    \]
    Thus, the member strings are accepted with probability 1 (i.e., perfect-completeness) in exponential expected time.

    If the input is not a member string, we know that in at least one of the universal choices, we have $S \neq T$. Then, we check the ratio of rejecting and accepting probabilities in a single iteration. If $S \neq T$, then $(S-T)^2$ can be at least 1. Then, this ratio cannot be less than
    \[
        \dfrac{1}{ 2^n \cdot (1/16)^n } = 8^n,
    \]
    where the acceptance probabilities come from all universal choices ($2^n$ in total).
    Thus, the non-member strings are rejected with probability at least
    \[
        \dfrac{8}{8+1} > 0.888.
    \]
    This probability can increased by tuning the acceptance probability. The running time is again exponential in the input length.
\end{proof}

\section{Simulating linear-space Turing machines}
\label{sec:linear-space-TMs}

\begin{fact}
    \label{fact:linear-DTM} \cite{Yak13C}
    Any language recognized by a linear-space DTM, say $M$, can be verified by a 2QCFA with perfect-completeness in double-exponential expected time.
\end{fact}
\begin{proof}
    Let $k$ be an integer such that the length of any configuration of $M$ (as explained in Sec.~\ref{sec:encode-conf-TM}) is at most $k \cdot n $ for every $n>0$, where $n$ is the length of input. We assume that the symbol $ @ $ is not in the configuration alphabet. For a given input $w$, a valid configuration history of $M$ on $w$ is 
    \[
        c_1\confsep c_2 \confsep c_3 \confsep \cdots \confsep c_f \confsep \dollar,
    \]
    where $c_1$ is the initial configuration, $c_i$ is the configuration obtained from $c_{i-1}$ in a single transition for $i>1$, and $c_f$ is the final configuration. Here we separate each configuration with four symbols, and we apply a superoperator for each.\footnote{One may indeed use less symbols here by combining some of superoperators appropriately.}

    The verifier runs an infinite loop and in each iteration requests a valid configuration history of $M$ on $w$. During an iteration, whenever an outcome associated to an axillary operation element is observed, it is terminated and a new iteration begins.

    \textbf{Deterministic checks.} 
    The verifier does certain deterministic checks in each iteration and rejects the input if any of them fails.
    \begin{enumerate}
        \item  The verifier deterministically checks if the first configuration provided by the prover is $c_1$ by reading the input once from the left to the right. 
        \item  The verifier uses its classical states to count the number of $@$ symbols and checks if it is equal to 4 in each block.
        \item  The verifier checks if the configuration history ends with a $\dollar$ symbol.
        \item The verifier checks if there is only one halting configuration and it is the last one.
        \item The verifier ensures that the length of each configuration provided by the prover is at most $k \cdot |w|$ by using its input head and classical states: the verifier uses its input marked with the end-markers as a counter of length $|w|$, and then, easily count up to $k \cdot |w|$ by staying $k$ steps on each input symbol.
    \end{enumerate}

    Thanks to these deterministic checks -- in the remaining part, we assume that the prover sends a configuration history as
    \[
        c_1 \confsep c'_2 \confsep c'_3 \confsep \cdots \confsep c'_{f'} \confsep \dollar.
    \]
    Otherwise, the input is rejected by the verifier.

    \textbf{Quantum checks.} The verifier checks the correctness of the configuration sequences provided by the prover. It uses a quantum register with four states: $\{q_1,q_2,q_3,q_4\}$.
    
    While scanning $c_1$, it encodes the value of the configuration obtained from $c_1$ in a single transition, namely $\nxt(c_1)$, into the amplitude of $q_4$. When reading the forth $@$ symbol, $\nxt(c_1)$ is transferred to the amplitude of $q_2$, and the amplitude of $q_4$ is set to 0. 
    
    Then, the verifier scans $c'_2 \confsep$. The value of $c'_2$ is encoded into the amplitude of $q_3$, and $\nxt(c'_2)$ is encoded into the amplitude of $q_4$. The quantum state is as below just before reading the last $@$ symbol:
    \[
        N_2 \myvec{1 \\ \nxt(c_1) \\ \val(c'_2) \\ \nxt(c'_2) },
    \]
    where the $N_2$ is the normalization factor. The next superoperator has two main operational elements. The quantum state obtained by the first one is
    \begin{equation}
        \label{eq:mismatch}
        N''_2 \myvec{ 3(\nxt(c_1) - \val(c'_2)) \\ 0 \\ 0 \\ 0 },
    \end{equation}
    where $N''_2$ is the normalization factor. If this is observed, the input is rejected by the verifier. It is easy to see that this can be observed only if $ \nxt(c_1) \neq \val(c'_2) $, i.e., the prover cheats about $c'_2$, which is different than $c_2$. This part is indeed how the verifier checks if the prover cheats about the configuration history. The quantum state obtained by the second operation element is 
    \[
        N'_2 \myvec{1 \\ \nxt(c'_2) \\ 0 \\ 0},
    \]
    where $N'_2$ is the normalization factor. The iteration continues if this is observed.

    When scanning $c'_3 \confsep$, the quantum state, just before reading the last $@$ symbol, is
    \[
        N_3 \myvec{1 \\ \nxt(c'_2) \\ \val(c'_3) \\ \nxt(c'_3) },
    \]
    where the $N_3$ is the normalization factor. Again the superoperator has two main operational elements. The quantum state obtained by the first one is
    \[
        N''_3 \myvec{3(\nxt(c'_2) - \val(c'_3)) \\ 0 \\ 0 \\ 0 },
    \]
    where $N''_3$ is the normalization factor. If this is observed, the input is rejected by the verifier. This can happen only if $ \nxt(c'_2) \neq \val(c'_3) $. The quantum state by the second operation element is 
    \[
        N'_3 \myvec{1 \\ \nxt(c'_3) \\ 0 \\ 0},
    \]
    where $N'_3$ is the normalization factor. The iteration continues if this is observed.
    
    The verifier scans the rest of configurations provided by the prover in the same way. 
    
    After scanning $c'_{f'} \confsep$, the quantum state is 
    \[
        N'_{f'} \myvec{1 \\ 0 \\ 0 \\ 0},
    \]
    where $N'_{f'}$ is the normalization factor. By scanning the $\dollar$ symbol provided by the prover, the verifier makes its decision based on the halting configuration $c'_{f'}$. If it is an accepting (resp., a rejecting) configuration, the verifier accepts (resp., rejects) the input.
    If the length of configuration history except $\dollar$ provided by the prover is $l$, then the probability of this decision is 
    \begin{equation}
        \label{eq:2l}
        \mypar{ \frac{1}{D}}^{2l}.
    \end{equation}

    If there is any mismatch in the configuration history, i.e., $\nxt(c_i) \neq \val(c_{i+1}) $ for some $i$, the input is rejected with probability at least
    \begin{equation}
        \label{eq:2lprime}
        9 \mypar{ \frac{1}{D}  }^{2l'},
    \end{equation}
    where $l' \leq l$ is the number of symbols provided by the prover until that point.

    If $w$ is in the language, the configuration history provided by an honest prover is valid. So, the input is never rejected but accepted with probability $\mypar{ \frac{1}{D}}^{2l}$ in a single iteration. Thus, the input is accepted with probability 1.

    If $w$ is not in the language, we have two cases. If the prover is honest, the input is never accepted but rejected with probability $\mypar{ \frac{1}{D}}^{2l}$ in a single iteration. If the prover is cheating, then there must be a configuration mismatch. Thus, the rejecting probability in this iteration is at least 9 times more than any possible accepting probability due to Eq.~\ref{eq:2l} and Eq.~\ref{eq:2lprime}. So, the overall rejecting probability is at least
    \[
        \frac{9}{9+1} = \frac{9}{10}.
    \]
    By using a coefficient greater than 3 in Eq.~\ref{eq:mismatch} (and in the other mismatch checks), we can obtain better rejecting probability.

    The minimal halting probability is obtained when receiving a valid configuration history. That is,  
    \[
        \mypar{ \frac{1}{D} }^{2l} = D^{-2l}.
    \]
    The value of $l$ can be bounded by $|w| \cdot 2^{O(|w|)}$, as the number of all possible configurations can be at most exponential in $|w|$ when using linear space. Thus, the expected running time is
    \[
        2^{2^{O(|w|)}},
    \]
    which is double-exponential in $|w|$.   
\end{proof}

\begin{fact}
    \label{fact:linear-NTM} \cite{Yak13C}
    Any language recognized by a linear-space NTM can be verified by a 2QCFA with perfect-completeness in double-exponential expected time.
\end{fact}
\begin{proof}
    We adopt the proof given for linear-space DTMs. The nondeterministic choices are selected by the prover. Before receiving a configuration $ c $ from the prover, the verifier asks for which nondeterministic choice is picked when in $c$. Based on the provided selection, the verifier prepares $\nxt(c)$ and compares it with the configuration provided by the prover immediately after $c$. The rest of the proof is the same. 

    For the member inputs, the honest prover sends a valid configuration history, and so, those inputs are accepted by the verifier with probability 1.

    For the non-member inputs, the verifier can make the decision of ``acceptance'' only if an invalid configuration history is provided by the prover. Thus, the rejecting probability due to a mismatch is always sufficiently bigger than the any potential accepting probability as explained before.
\end{proof}

\begin{theorem}
    Any language $L$ linear-space reducible to $\knapsackgame$ is verified by a 2QCFA verifier, say $V$, with perfect completeness. The protocol runs in double-exponential expected time.
\end{theorem}
\begin{proof}
    Let $M$ be a linear-space DTM reducing $L$ to $\knapsackgame$. A reducer may have a read-only tape and a working tape separately, especially when focusing on sub-linear space (e.g., log-space) reductions. But, in our case, $M$ has a single read/write tape. A reducer uses a write-only output tape to write the output string. However, in our simulation below, we assume that $M$ does not have any output tape. We use the same configuration history format presented in the proof of Fact~\ref{fact:linear-DTM} for $M$. While it scans any configuration of $M$, the verifier knows the next symbol to be written on the output tape. To process this symbol, it uses its internal states to store it temporarily. We assume that once $M$ completes its computation, it enters an accepting state, which is the only halting state.
    
    Our verifier $V$ executes the 2QCFA verifiers given in the proofs of Fact~\ref{fact:linear-DTM} and Theorem~\ref{thm:knapsack-game} in parallel. We name them $V_1$ and $V_2$, respectively. 
    \begin{itemize}
        \item The verifier $V$ requests from the prover the valid configuration history of $M$ on the given input.
        \item The verifier $V_1$ ensures that the provided configuration sequence is valid or not.
        \item While reading every configuration, $V$ checks if $M$ outputs any symbol when in this configuration. If so, $V$ stores this symbol.
        \item Once a new symbol is stored by $V$, it is fed to $V_2$.
    \end{itemize}

     The quantum register of $V$ is combined by the quantum registers of $V_1$ and $V_2$. But, each sub-register works separately. So, we can see them as two separate quantum registers.

    The input is accepted by $V$ if both $V_1$ and $V_2$ gives the decision of ``acceptance''. Otherwise, the input is rejected by $V$. In this way, the perfect-completeness property is preserved.

    Let $w$ be a given input, and let $u$ be the output string by $M$ when starting with $w$. 
    
    \textbf{Member inputs ($w \in L$).} By communicating with an honest prover, $V$ obtains $u$ symbol by symbol and then feed it to $V_2$. We know that $u$ is a member string of $\knapsackgame$. As the computation history of $M$ is valid, $V_1$ only makes the decision of ``acceptance''. Similarly, the honest prover provides valid nondeterministic choices to $V_2$ when processing $u$. Then, $V_2$ also only makes the decision of ``acceptance''. Thus, $w$ is accepted by $V$ with probability 1.

    \textbf{Non-member inputs ($w \notin L$).} For simplicity, we assume that $V$ communicates with two provers: $V_1$ communicates with $P_1$, and $V_2$ communicates with $P_2$. (There is a single prover, but in this way, we can separate the communications of $V_1$ and $V_2$ from each other.)
    \begin{itemize}
        \item If $P_1$ is honest, $V_1$ only accepts $w$, and $u \notin \knapsackgame $ is fed to $V_2$. Then, we know that $V_2$ rejects $u$ with high probability independent of the honesty of $V_2$. Thus, $V$ rejects $w$ with high probability.
        \item If $P_1$ is not honest, then $V_1$ rejects $w$ with high probability. So, even if $V_2$ only accepts $u$, $V$ rejects $w$ with high probability.
    \end{itemize}
The verifier $V_1$ runs in double-exponential expected time. The (expected) length of $u$ can be in exponential in $|w|$, and so, $V_2$ also runs in double-exponential expected time. Both $V_1$ and $V_2$ runs in parallel, so, $V$ runs in double-exponential expected time. 
\end{proof}

\begin{corollary}
    Any language is PSPACE is verified by a 2QCFA verifier with perfect-completeness. The protocol runs in double-exponential time.
\end{corollary}
\begin{proof}
    This follows from the fact that any language in PSPACE is log-space reducible to $\knapsackgame$. Remark that any log-space reduction can be exactly simulated by a linear-space reducer without using any extra work tape.
\end{proof}

\section{Concluding Remarks}
\label{sec:concluding-remarks}

Arthur-Merlin (AM) systems with 2PFA verifiers are too weak when comparing with AM systems with 2QCFA verifiers. Comparing them with private-coin interactive proof systems (IPSs) with 2PFAs verifiers makes more sense. We observe that being in superposition makes the protocols stronger similar to private-coin protocols. In addition to that 2QCFAs benefit from interference (see \cite{AW02,YS10B}). 

It is open if IPSs with 2PFA verifiers can go beyond $\mathsf{ASPACE(n)} = \mathsf{DTIME\left(2^{O(n)}\right)}$ or if they can verify any language in $ \mathsf{PSPACE} $ (or $\mathsf{NP}$). The class $\mathsf{AM_1(2QCFA)}$ already contains $\mathsf{ASPACE(n)}$, $\mathsf{PSPACE} = \mathsf{QIP}$, and an $\mathsf{NEXP}$-complete language. One may ask if $\mathsf{AM_1(2QCFA)}$ or $\mathsf{AM(2QCFA)}$ contains $\mathsf{EXP}$ or $\mathsf{NEXP}$ (or beyond). In such case, the verification power of 2QCFAs can be compared with polynomial-time multi-prover interactive proof systems.

\section*{Acknowledgments}

Yakary{\i}lmaz was partially supported by the Latvian Quantum Initiative under European Union Recovery and Resilience Facility project no. 2.3.1.1.i.0/1/22/I/CFLA/001

\bibliographystyle{quantum}
\bibliography{ref}

\begin{thebibliography}{10}

\bibitem{Co93A}
Anne Condon.
\newblock ``Complexity theory: Current research''.
\newblock Chapter The complexity of space bounded interactive proof systems, pages 147--190.
\newblock Cambridge University Press. ~(1993).
\newblock  url:~\url{https://dl.acm.org/doi/10.5555/183589.183728}.

\bibitem{Wa09}
John Watrous.
\newblock ``Quantum computational complexity''.
\newblock In Robert~A. Meyers, editor, Encyclopedia of Complexity and Systems Science.
\newblock \href{https://dx.doi.org/10.1007/978-0-387-30440-3\_428}{Pages 7174--7201}.
\newblock Springer~(2009).

\bibitem{GMR85}
Shafi Goldwasser, Silvio Micali, and Charles Rackoff.
\newblock ``The knowledge complexity of interactive proof-systems (extended abstract)''.
\newblock In STOC'85: Proceedings of the 17th Annual ACM Symposium on Theory of Computing.
\newblock \href{https://dx.doi.org/10.1145/22145.22178}{Pages 291--304}.
\newblock ~(1985).

\bibitem{GMR89}
Shafi Goldwasser, Silvio Micali, and Charles Rackoff.
\newblock ``The knowledge complexity of interactive proof systems''.
\newblock \href{https://dx.doi.org/10.1137/0218012}{SIAM Journal on Computing {\bf 18}, 186--208}~(1989).

\bibitem{Ba85}
L{\'a}szl{\'o} Babai.
\newblock ``Trading group theory for randomness''.
\newblock In STOC'85: Proceedings of the 17th Annual ACM Symposium on Theory of Computing.
\newblock \href{https://dx.doi.org/10.1145/22145.22192}{Pages 421--429}.
\newblock ~(1985).

\bibitem{Wat99B}
John Watrous.
\newblock ``\mbox{PSPACE} has constant-round quantum interactive proof systems''.
\newblock In FOCS'99: Proceedings of the 40th Annual Symposium on Foundations of Computer Science.
\newblock \href{https://dx.doi.org/10.1109/SFFCS.1999.814583}{Pages 112--119}.
\newblock ~(1999).

\bibitem{DS92}
Cynthia Dwork and Larry Stockmeyer.
\newblock ``Finite state verifiers $\mbox{I}$: The power of interaction''.
\newblock \href{https://dx.doi.org/10.1145/146585.146599}{Journal of the ACM {\bf 39}, 800--828}~(1992).

\bibitem{Co91}
Anne Condon.
\newblock ``Space-bounded probabilistic game automata''.
\newblock \href{https://dx.doi.org/10.1145/103516.128681}{Journal of the ACM {\bf 38}, 472--494}~(1991).

\bibitem{Yak12A}
Abuzer Yakary{\i}lmaz.
\newblock ``Public-qubits versus private-coins''.
\newblock Electron. Colloquium Comput. Complex.{\bf {TR12-130}}~(2012).
\newblock  url:~\url{https://eccc.weizmann.ac.il/report/2012/130}.

\bibitem{Yak13C}
Abuzer Yakary{\i}lmaz.
\newblock ``Public qubits versus private coins''.
\newblock In The Proceedings of Workshop on Quantum and Classical Complexity.
\newblock Pages 45--60.
\newblock {}~(2013). Univeristy of Latvia Press.
\newblock  url:~\url{https://users.utu.fi/mikhirve/workshop/Proceedings.pdf}.

\bibitem{Ay15}
Andris Ambainis and Abuzer Yakary{\i}lmaz.
\newblock ``Automata: From mathematics to applications''.
\newblock Technical Report 1507.01988.
\newblock arXiv~(2015).
\newblock  url:~\url{http://arxiv.org/abs/1507.01988}.

\bibitem{AY21}
Andris Ambainis and Abuzer Yakary{\i}lmaz.
\newblock ``Automata and quantum computing''.
\newblock In Jean{-}{\'{E}}ric Pin, editor, Handbook of Automata Theory.
\newblock \href{https://dx.doi.org/10.4171/AUTOMATA-2/17}{Pages 1457--1493}.
\newblock European Mathematical Society Publishing House, Z{\"{u}}rich, Switzerland, {}~(2021).

\bibitem{AW02}
Andris Ambainis and John Watrous.
\newblock ``Two--way finite automata with quantum and classical states''.
\newblock \href{https://dx.doi.org/10.1016/S0304-3975(02)00138-X}{Theoretical Computer Science {\bf 287}, 299--311}~(2002).

\bibitem{YS11A}
Abuzer Yakary{\i}lmaz and A.~C.~Cem Say.
\newblock ``Unbounded-error quantum computation with small space bounds''.
\newblock \href{https://dx.doi.org/10.1016/J.IC.2011.01.008}{Information and Computation {\bf 279}, 873--892}~(2011).

\bibitem{SY17}
A.~C.~Cem Say and Abuzer Yakary{\i}lmaz.
\newblock ``Magic coins are useful for small-space quantum machines''.
\newblock \href{https://dx.doi.org/10.26421/QIC17.11-12-6}{Quantum Inf. Comput. {\bf 17}, 1027--1043}~(2017).

\bibitem{BV97}
Ethan Bernstein and Umesh Vazirani.
\newblock ``Quantum complexity theory''.
\newblock \href{https://dx.doi.org/10.1137/S0097539796300921}{SIAM Journal on Computing {\bf 26}, 1411--1473}~(1997).

\bibitem{Wat03}
John Watrous.
\newblock ``On the complexity of simulating space-bounded quantum computations''.
\newblock \href{https://dx.doi.org/10.1007/S00037-003-0177-8}{Computational Complexity {\bf 12}, 48--84}~(2003).

\bibitem{Travers06}
Stephen~D. Travers.
\newblock ``The complexity of membership problems for circuits over sets of integers''.
\newblock \href{https://dx.doi.org/10.1016/J.TCS.2006.08.017}{Theor. Comput. Sci. {\bf 369}, 211--229}~(2006).

\bibitem{FearnleyJ13}
John Fearnley and Marcin Jurdzi\'{n}ski.
\newblock ``Reachability in two-clock timed automata is {PSPACE}-complete''.
\newblock In Fedor~V. Fomin, Rusins Freivalds, Marta~Z. Kwiatkowska, and David Peleg, editors, Automata, Languages, and Programming - 40th International Colloquium, {ICALP} 2013, Riga, Latvia, July 8-12, 2013, Proceedings, Part {II}.
\newblock \href{https://dx.doi.org/10.1007/978-3-642-39212-2\_21}{Volume 7966 of Lecture Notes in Computer Science, pages 212--223}.
\newblock Springer~(2013).

\bibitem{FearnleyJ15}
John Fearnley and Marcin Jurdzi\'{n}ski.
\newblock ``Reachability in two-clock timed automata is {PSPACE}-complete''.
\newblock \href{https://dx.doi.org/10.1016/J.IC.2014.12.004}{Inf. Comput. {\bf 243}, 26--36}~(2015).

\bibitem{YS10B}
Abuzer Yakary{\i}lmaz and A.~C.~Cem Say.
\newblock ``Succinctness of two-way probabilistic and quantum finite automata''.
\newblock \href{https://dx.doi.org/10.46298/DMTCS.509}{Discrete Mathematics and Theoretical Computer Science {\bf 12}, 19--40}~(2010).

\end{thebibliography}

\end{document}